\documentclass[12pt]{article}

\usepackage{times}
\usepackage{fullpage}
\usepackage{amsfonts,amssymb,amsmath}
\usepackage{latexsym}
\usepackage{url}

\newcommand{\R}{\mathbb{R}}

\newcommand{\floor}[1]{\lfloor #1 \rfloor}
\newcommand{\ceil}[1]{\left\lceil #1 \right\rceil}

\def\01{\{0,1\}}
\def\mp{\{-1,1\}}

\newcommand{\sn}[1]{\|#1\|}

\newcommand{\Th}{\mathrm{Th}}

\newcommand{\OR}{\mathrm{OR}}
\newcommand{\AND}{\mathrm{AND}}
\newcommand{\Ex}{\mathbb{E}}

\newcommand{\DISJ}{\mathrm{DISJ}}
\newcommand{\NOF}{number-on-the-forehead }
\newcommand{\LS}{Lov{\'a}sz-Schrijver }

\newcommand{\braket}[2]{\langle#1, #2\rangle}

\newcommand{\disc}{\mathrm{disc}}

\newcommand{\ACC}{\mathrm{ACC}}
\newcommand{\AC}{\mathrm{AC}}
\newcommand{\size}{\mathrm{size}}

\newcommand{\ignore}[1]{}

\newtheorem{theorem}{Theorem}
\newtheorem{lemma}[theorem]{Lemma}
\newtheorem{corollary}[theorem]{Corollary}

\newtheorem{definition}[theorem]{Definition}
\newtheorem{remark}[theorem]{Remark}

\newtheorem{fact}[theorem]{Fact}

\newcommand{\hyperref}[2][]{#2}

\newcommand{\reff}[1]{\ref{#1}}
\newcommand{\thmref}[1]{\hyperref[#1]{{Theorem~\reff{#1}}}}
\newcommand{\lemref}[1]{\hyperref[#1]{{Lemma~\reff{#1}}}}
\newcommand{\corref}[1]{\hyperref[#1]{{Corollary~\reff{#1}}}}
\newcommand{\eqnref}[1]{\hyperref[#1]{{Equation~(\reff{#1})}}}
\newcommand{\factref}[1]{\hyperref[#1]{{Fact~\reff{#1}}}}
\newcommand{\defref}[1]{\hyperref[#1]{{Definition~\reff{#1}}}}
\newcommand{\secref}[1]{\hyperref[#1]{{Section~\reff{#1}}}}

\newenvironment{proof}
{\noindent {\bf Proof:}}
{{\hfill $\Box$}\\
 \smallskip}

\begin{document}
\title{Disjointness is hard in the multiparty number-on-the-forehead model}
\author{Troy Lee \\ Department of Computer Science \\ Columbia University
\thanks{Work conducted at Rutgers University, supported by a NSF Mathematical Sciences
Postdoctoral Fellowship.  Email: troyjlee@gmail.com}
\and Adi Shraibman \\ Department of Mathematics \\ Weizmann Institute of Science
\thanks{Email: adi.shraibman@weizmann.ac.il}}
\date{}
\maketitle

\begin{abstract}
We show that disjointness requires randomized communication
$\Omega \left(\frac{n^{1/(k+1)}}{2^{2^{k}}}\right)$ in the general
$k$-party number-on-the-forehead model of complexity.  The
previous best lower bound for $k \ge 3$ was
$\frac{\log n}{k-1}$ .  Our results give a separation between nondeterministic and
randomized multiparty \NOF communication complexity for up to
$k=\log \log n - O(\log \log \log n)$ many players.
Also by a reduction of Beame, Pitassi, and Segerlind,
these results imply subexponential lower bounds on the size of proofs needed
to refute certain unsatisfiable CNFs in a broad class of proof
systems, including tree-like \LS proofs.
\end{abstract}

\section{Introduction}
Since its introduction thirty years ago \cite{Abel, Yao79}, communication
complexity has become a key concept in complexity theory and
theoretical computer science in general. Part of its appeal is
that it has applications to many different computational models,
for example to formula size and circuit depth, proof complexity,
branching programs, VLSI design, and time-space trade-offs for
Turing machines (see \cite{KN97} for more details).

One area of communication complexity which still holds many
mysteries is the $k$-party ``number-on-the-forehead'' model,
originally introduced by Chandra, Furst, and Lipton \cite{CFL83}.
In this model, $k$ parties wish to compute a function $f :
(\{-1,+1\}^n)^k \rightarrow \mp$.  On input $(x_1,\ldots,x_k)$,
the $i^{th}$ player receives $(x_1, \ldots, x_{i-1}, x_{i+1},
\ldots, x_k)$. That is, player $i$ has knowledge of the entire
input {\em except} for the string $x_i$, which figuratively can be
thought of as sitting on his forehead. The players communicate by
writing messages ``on a blackboard,'' so that all players see each
message. The large overlap in the player's knowledge is part of
what makes showing lower bounds in this model so difficult. This
difficulty, however, is rewarded by the richness and strength of
consequences of such lower bounds: for example, by results of
\cite{HG91, BT94}, showing a super-polylogarithmic lower bound on
an explicit function for polylogarithmic many players would give
an explicit function outside of the class $\ACC^0$ --- that is, a
function which requires super-polynomial size constant-depth
circuits using AND, OR, NOT, and modulo $m$ gates.

While showing such bounds remains a challenging open problem, we
do know of explicit functions which require large communication in
this model for $\Theta(\log n)$ many players. Babai,
Nisan, and Szegedy \cite{BNS89} showed that the inner product
function generalized to $k$-parties requires randomized
communication $\Omega(n/4^k)$, and for other explicit functions
slightly larger bounds of size $\Omega(n/2^k)$ are known
\cite{FG05}. These lower bounds are all achieved using the
discrepancy method, a very general technique which gives lower
bounds even on randomized models with error probability close to
$1/2$, and also on nondeterministic communication complexity.

For some basic functions, however, there is a huge gap in our knowledge. One
example is the disjointness function, or equivalently its
complement, set intersection.  In the set intersection problem,
the goal of the players is to determine if there is an index $j$
such that every string $x_i$ has a $-1$ in position $j$, where here and throughout
the paper we interpret $-1$ as `true.'  The best
known protocol has cost $O(k^2 n \log(n)/2^k)$ \cite{Gro94}.  On the
other hand, the best lower bound in the general \NOF model
is $\tfrac{\log n}{k-1}$, for $k \ge 3$ \cite{Tes02, BPSW06}.
For $k=2$ tight bounds are known of $\Theta(n)$ for randomized communication
complexity \cite{KS87} and $\Theta(\sqrt{n})$ for quantum
communication complexity \cite{Raz03, AA05}.

A major obstacle toward proving better lower bounds on set
intersection is that it has a low cost nondeterministic
protocol.  In case there is a position where all players have a $-1$, with 
$O(\log n)$ bits a prover can send the name of this position
and the players can then verify
this is the case.  Since the discrepancy method is also a lower bound on
nondeterministic complexity, it is limited to logarithmic
lower bounds for set intersection.  Even in the two-party case,
determining the complexity of set intersection in the randomized
and quantum models was a long-standing open problem, in part for this reason.

In the multiparty case, the discrepancy method is the only technique
which has been used to show lower bounds on the general randomized model of
\NOF complexity.  Although other two-party methods can be generalized to the
multiparty \NOF model, they can become very difficult to handle.
One source of this difficulty is that, whereas in the two party case we can
nicely represent the function $f(x,y)$ as a matrix, in the
multiparty case we deal with higher dimensional tensors.  This
makes many of the linear algebraic tools so useful in the two-party case
inapplicable or at least much more involved.  For example, while matrix
rank is a staple lower bound technique for deterministic two-party complexity,
in the tensor case even basic questions like the maximum rank of a
$n \times n \times n$ tensor remain open.

Besides this technical challenge, additional motivation
to studying the \NOF complexity of disjointness was given by Beame, Pitassi,
and Segerlind \cite{BPS06}, who showed that lower bounds on disjointness imply
lower bounds on a very general class of proof systems, including
cutting planes and \LS proof systems.

We show that disjointness requires randomized communication
$\Omega \left(\frac{n^{1/(k+1)}}{2^{2^k}}\right)$ in the general
$k$-party \NOF model.  This separates nondeterministic and randomized
multiparty \NOF complexity for up to $k=\log \log n - O(\log \log \log n)$ many
players.  Also by the work of \cite{BPS06} this implies subexponential lower
bounds on the size of proofs needed to refute certain unsatisfiable formulas
by tree-like proofs in \LS and more powerful proof systems.

Chattopadhyay and Ada \cite{CA08} have independently obtained similar
bounds on disjointness using similar techniques.

\subsection{Related work}
For restricted models of computation, bounds are known which are
stronger than ours.  Wigderson showed that for one-way three-party
\NOF protocols, disjointness requires communication
$\Omega(n^{1/2})$ (this result appears in \cite{BHK01}). More
recently, Viola and Wigderson \cite{VW07b} extended this approach
to show a bound of  $\Omega(n^{1/(k-1)}/k^{O(k)})$ on the
complexity of one-way $k$-party protocols computing disjointness.
These results actually show bounds on a pointer jumping function
which reduces to disjointness.

Beame, Pitassi, Segerlind, and Wigderson \cite{BPSW06} devised a
method based on a direct product theorem to show a
$\Omega(n^{1/3})$ bound on the complexity of three-party
disjointness in a model stronger than one-way where the first
player speaks once, and then the two remaining players interact
arbitrarily.

Following up on our work, David, Pitassi, and Viola \cite{DPV08}
gave an explicit function which separates nondeterministic and
randomized \NOF communication complexity for up to $\Omega(\log
n)$ players.  They are also able, for any constant $c$ to give a
function computable in $\AC^0$ which separates them for up to $c
\log \log n$ players.  Note that disjointness can be computed in
$\AC^0$, but that our bounds are already trivial for $\log \log n$
players.  Even more recently, Beame and Huynh-Ngoc \cite{BHN08}
have shown a bound of $2^{\Omega(\sqrt{\log n}/\sqrt{k})-k}$ on the
$k$-party \NOF complexity of disjointness.  This bound remains
non-trivial for up to $\Theta(\log^{1/3} n)$ many players, but is
not as strong as our bound for few players.

\subsection{Overview of techniques}
There is a natural correspondence between functions $f :
(\{-1,+1\}^n)^k \rightarrow \mp$ and sign $k$-tensors. Sometimes
it is more convenient to consider the function form, and sometimes,
like when discussing norms, it is more convenient to consider
tensors.

Our proof combines two ingredients.  The first of these is the
notion of an approximation norm.  For a norm $\Phi$, and a sign
tensor $A$, the {\em approximation norm} associated to $\Phi$ and
$A$, denoted $\Phi^\alpha(A)$, is the smallest $\Phi$ norm of an
element `close' to $A$.  Here $\alpha$ quantifies the term
`close.'

Approximation norms turn out to be quite useful for showing lower bounds
on randomized and quantum communication complexity \cite{Kla01, Raz03, LS07}.
Razborov, for example, uses the approximation trace norm to
prove a tight lower bound on the quantum communication complexity of set
intersection.

We use what we call the {\em cylinder
intersection norm}, denoted $\mu$.
This norm can be seen as a multiparty generalization of a quantity
used in Lemma~3.1 of Klauck \cite{Kla01}.  As a correct deterministic protocol
partitions the communication matrix into rectangles
on which the function is constant, analogously a correct
deterministic \NOF protocol decomposes the
communication tensor into cylinder intersections on which the
function is constant.  Roughly speaking, $\mu(A)$ measures how
efficiently $A$ can be written as a sum of cylinder intersections.  In this way,
if $A$ has low communication complexity, it will also have low $\mu$ norm.
We defer formal definitions to \secref{sec:the_method}.

We denote the approximate version of the cylinder intersection
norm by $\mu^{\alpha}$ where $1 \le \alpha < \infty$ represents
the measure of approximation. This measure provides a lower bound on
randomized communication complexity in the \NOF model.
The limiting case $\mu^\infty(A)$ turns out to be exactly the
usual discrepancy method.  For bounded $\alpha$ we obtain a
technique which is strictly stronger than the discrepancy
method.

Following \cite{LMSS07, LS07}, to show lower bounds on $\mu^\alpha(A)$, 
we write it in terms of the dual norm $\mu^*$.  By definition of a dual norm, we have
\begin{equation}
\mu(B)=\max_Q \frac{\braket{B}{Q}}{\mu^*(Q)}.
\label{dual}
\end{equation}
This ``max'' formulation of $\mu$ is often more convenient for showing 
lower bounds.  The dual norm $\mu^*$ is closely related to discrepancy 
with respect to the uniform distribution, so we can use existing techniques 
to upper bound $\mu^*(Q)$.  

This formulation of $\mu$ also gives a way to write $\mu^\alpha$ 
in terms of a maximization quantity.
\begin{equation}
\mu^\alpha(A)=\max_Q
\frac{(1+\alpha)|\braket{A}{Q}|+(1-\alpha)\|Q\|_1}{2 \mu^*(Q)}.
\label{approx_dual}
\end{equation}
All one needs for showing lower bounds is that the left hand side is at 
least as large as the right hand side.  This can be shown quite simply 
using Equation~\ref{dual} and elementary inequalities and was noted, for 
example, by Razborov in the context of the approximation trace norm.  The fact 
that equality holds here requires the use of linear programming 
duality or a separation theorem for convex bodies and seems to 
be less well known.  

As the dual norm $\mu^*$ is essentially discrepancy with respect
to the uniform distribution, the approximation $\mu$ norm can
be seen as an extension of discrepancy in another way.
Instead of proving that the tensor of interest $A$ has small discrepancy, 
it is enough to prove that there is a tensor $Q$ which has small discrepancy
and has large correlation with $A$, relative to $\|Q\|_1$.  This is why this 
method is called {\em generalized discrepancy} in \cite{CA08}.  

To find a good witness tensor $Q$, we use ideas from a second line
of research.  While the norm framework of
\eqnref{approx_dual} provides a nice approach to lower bound
communication complexity, it gives no hint about how to choose a
good witness $Q$---in general a difficult problem. Works by
Sherstov \cite{She07, She08a} and Shi and Zhu \cite{SZ07} in the
two-party case, and Chattopadhyay \cite{Cha07} in the multiparty
case provide an elegant way to choose a good witness for a general
class of matrices and tensors.  These works look at block composed 
functions of the form $f \circ g^n(x_1, \ldots, x_k)=
f(g(x_1^1, \ldots, x_k^1), \ldots, g(x_1^n, \ldots, x_k^n))$.
Notice that set intersection is a block composed function where $f=\OR_n$ is the 
OR function on $n$ bits
and $g=\AND_k$ is the $k$-player AND function on one bit.  Sherstov \cite{She07}
first showed that when $g(x,i)=x_i$, the discrepancy of a block 
composed function could be bounded in terms of the threshold degree of $f$, the 
minimum degree of a polynomial which agrees in sign with $f$ on the Boolean cube.  
Building on this result, Chattopadhyay showed an analogous statement in the 
\NOF case for an appropriately generalized multiparty function $g$.  

Sherstov and independently Shi-Zhu showed that the approximate trace norm 
of a block composed function could be lower bounded in terms of the approximate
degree of $f$, again provided that the inner function $g$ satisfies certain 
technical conditions.  The $\mu$ norm provides bounds at least as large as the 
trace norm method \cite{LS07}, thus these works also lower bound $\mu^\alpha$.  
In this paper, we take the natural step to show that $\mu^\alpha$ of 
a block composed multiparty function can be lower bounded in terms of the approximate 
degree of $f$, for a particular multiparty inner function $g$ such that the composed function
$f \circ g^n$ can be embedded in the set intersection problem.    

\subsection{Consequences for \LS proof systems and beyond}
\label{sec:bps} 
Beame, Pitassi, and Segerlind \cite{BPS06} show that bounds on
multiparty disjointness imply strong lower bounds on the size of refutations
of certain unsatisfiable formulas, for a very general class of
proof systems. We now introduce and motivate the study of these
proof systems. Formal definitions and the implications of our
results will be given in \secref{sec:proof_systems}.

The fact that linear and semidefinite programs can be solved with high precision 
in polynomial time is a remarkable algorithmic achievment.  It is
thus interesting to ask how these algorithms fare when pitted against
NP-complete problems.  For many NP-complete problems, there is a
very natural approach to solving them via linear or semidefinite
programming: namely, we first formulate the problem as optimizing
a convex function over the Boolean cube, i.e. with variables
subject to the quadratic constraints $x_i^2=x_i$.  We then relax
these quadratic constraints to linear or semidefinite constraints
to obtain a program which can be solved in polynomial time.  For
example, a linear relaxation of $x_i^2=x_i$ may simply be the
constraint $0 \le x_i \le 1$.  In the case of vertex cover, for
example, such a simple relaxation already gives a linear program
with approximation ratio of 2.  Semidefinite constraints are in
general more complicated, but there are several ``automatic'' ways
of generating valid semidefinite inequalities---that is,
semidefinite inequalities satisfied by all Boolean solutions of
the original problem.  Perhaps the best known of these is the \LS
``lift and project'' method \cite{LS91}.  The seminal
$0.878$-approximation algorithm for MAXCUT of Goemans and
Williamson \cite{GW95} can be obtained by relaxing the natural
Boolean programming problem with semidefinite constraints obtained
by one application of the \LS method.

As these techniques have given impressive results in approximation
algorithms, it is natural to ask if they can also be used to
efficiently obtain exact solutions.  Namely, how many inequalities
need to be added in general until all fractional optima are
eliminated and only true Boolean optima remain?

One way to address this question is to consider proof systems with
derivation rules based on linear programming or the \LS method.
Our particular application will look at the size of proofs needed
to refute unsatisfiable formulas.  Given a CNF $\phi$,  we can
naturally represent the satisfiability of $\phi$ as the
satisfiability of a system of linear inequalities, one for each
clause. For example, the clause $x_1 \vee x_4 \vee \neg x_5$ would
be represented as $x_1+x_4+(1-x_5) \ge 1$.  Suppose that $\phi$ is
unsatisfiable.  Then consider a proof system in which the
``axioms'' are the inequalities obtained from the clauses of
$\phi$, and the goal is to derive the contradiction $0 \ge 1$.  By
the results of \cite{BPS06}, our results on disjointness imply
that there are unsatisfiable formulas such that any refutation
obtained by generating new inequalities by the \LS method in a
``tree-like'' way requires size $2^{n^{\Omega(1)}}$.  For a
standard formulation of the \LS method known as $\mathrm{LS}_+$,
bounds of size $2^{\Omega(n)}$ for tree-like proofs have already
been shown by very different methods \cite{IK06}.

The advantage of the \NOF communication complexity approach,
however, is that it can also be applied to much more powerful
proof systems which are currently untouchable by other methods.
Beame, Pitassi, and Segerlind \cite{BPS06} show that  lower bounds
on $k$-party communication complexity of disjointness give lower
bounds on the size of tree-like proofs of certain unsatisfiable
CNFs $\phi(x)$, where the derivation rule is as follows: from
inequalities $f,g$ of degree $k-1$ in $x$, we are allowed to
conclude a degree $k-1$ inequality $h$ if every Boolean assignment
to $x$ which satisfies $f$ and $g$ also satisfies $h$. \LS proof
systems are a special case of such degree-2 systems.  Our bounds
on disjointness imply the existence of unsatisfiable formulas
whose refutation requires subexponential size tree-like degree-$k$
proofs, for any constant $k$. \footnote{The conference version of
this paper reported bounds on degree-$k$ proof systems for up to
$k=\log \log n -O(\log \log \log n)$.  As pointed out to us by
Paul Beame, however, this is not justified by the reduction of
\cite{BPS06}, which requires certain constraints on the size of
$k$.}  The aforementioned lower bounds on $\mathrm{LS}_+$ proof
systems strongly rely on specific properties of the \LS
operator---showing superpolynomial bounds on the size of tree-like
proofs in the more general degree-$k$ model was previously open
even in the case $k=2$.

\section{Preliminaries and notation}
\label{sec:prelim:not}

We let $[n]=\{1, \ldots, n\}$. For multiparty communication
complexity it is convenient to work with tensors, the
generalization of matrices to higher dimensions.  If an element of
a tensor $A$ is specified by $k$ indices, we say that $A$ is a $k$-tensor.
For a $k$-tensor $A$ of dimensions $(n_1,
\ldots, n_k)$ we write $\size(A)=n_1\cdots n_k$.  A tensor for which
all entries are in $\mp$ we call a sign tensor. For a function $f:
X_1 \times \ldots \times X_k \rightarrow \mp$, we define the
communication tensor corresponding to $f$ to be a $k$-tensor
$A_f$ where $A_f[x_1, \ldots, x_k]=f(x_1, \ldots, x_k)$.  We
identify $f$ with its communication tensor. For a set $Z \subseteq
X_1 \times \ldots \times X_k$ we let $\chi(Z)$ be its
characteristic tensor where $\chi(Z)[x_1, \ldots, x_k]=1$ if
$(x_1, \ldots, x_k) \in Z$ and is $0$ otherwise.

For a sign tensor $A$, we denote by $D^k(A)$ the deterministic
communication complexity of $A$ in the $k$-party
number-on-the-forehead model. The public coin randomized
communication complexity with error bound $\epsilon \ge 0$ is
denoted $R_{\epsilon}^k(A)$. We drop the superscript when the
number of players is clear from context.

We use the shorthand $A \ge c$ to indicate that all of the entries
of $A$ are at least $c$. The Hadamard or entrywise product of two
tensors $A$ and $B$ is denoted by $A \circ B$.  Their inner
product is denoted $\braket{A}{B}=\sum_{x_1, \ldots, x_k} A[x_1,
\ldots, x_k] B[x_1, \ldots, x_k]$. The $\ell_1$ and
$\ell_{\infty}$ norms of a tensor $A$ are $\|A\|_1=\sum_{x_1,
\ldots, x_k} |A[x_1, \ldots, x_k]|$ and $\|A\|_{\infty}=\max_{x_1,
\ldots, x_k} |A[x_1, \ldots, x_k]|$, respectively.

We also need some basic elements of Fourier analysis.  For $S
\subseteq [n]$ we define $\chi_S : \{-1,+1\}^n \rightarrow \mp$ as
$\chi_S(x)=\prod_{i \in S} x_i$.  As the $\chi_S$ form an
orthogonal basis, for any function $f:\{-1,+1\}^n \rightarrow \R$
we have a unique representation
$$
f(x)=\sum_{S \subseteq n} \hat f(S) \chi_S(x)
$$
where $\hat f(S)=(1/2^n)\braket{f}{\chi_S}$, are the Fourier
coefficients of $f$. The degree of $f$ is the size of the largest
set $S$ for which $\hat f(S)$ is nonzero.

\section{The Method}
\label{sec:the_method}

In this section we present a method for proving lower bounds on
randomized communication complexity in the number-on-the-forehead
model that generalizes and significantly strengthens the
discrepancy method.

\subsection{Cylinder intersection norm}
\label{sec:cylinder_norm}
 In two-party communication complexity, a key role is played by combinatorial rectangles---subsets
 of the form $Z_1 \times Z_2$ where $Z_1$ is a subset of inputs to Alice and $Z_2$ is a subset of
 inputs to Bob.  The analogous concept in the \NOF model of multiparty
communication complexity is that of a cylinder intersection.

\begin{definition}[Cylinder intersection]
A subset $Z_i \subseteq X_1 \times \ldots \times X_k$ is called a cylinder in the $i^{th}$ dimension
if membership in $Z_i$ does not depend on the $i^{th}$ coordinate.  That is,
for every $(z_1, \ldots, z_i, \ldots, z_k) \in Z_i$ and $z_i' \in X_i$ it also holds that
$(z_1, \ldots, z_i', \ldots, z_k) \in Z_i$.  A set $Z$ is called a cylinder intersection if it can be
expressed as $Z=\cap_{i=1}^k Z_i$ where each $Z_i$ is a cylinder in the $i^{th}$ dimension.
\end{definition}

Cylinder intersections are important because a correct
deterministic \NOF protocol for a function $f$ partitions the
corresponding communication tensor into cylinder intersections,
each of which is monochromatic with respect to the function $f$.

\begin{fact}
\label{fact:partition_det} Let $A$ be a sign $k$-tensor, and
suppose that $D^k(A) \le c$. Then there are cylinder intersections
$Z_1, \ldots, Z_{2^c}$ such that
$$
A=\sum_{i=1}^{2^c} \alpha_i \chi(Z_i)
$$
where $\alpha_i \in \{-1,+1\}$.
\end{fact}

Our main object of study, termed the cylinder intersection norm, relaxes this notion of
decomposition to allow $\alpha_i \in \R$.  A similar such relaxation is done by
\cite{KKN95} in the context of nondeterministic communication complexity.
\paragraph{Cylinder intersection norm}
We denote by $\mu$ the norm induced by the absolute convex hull of
the characteristic functions of all cylinder intersections.  That
is, for a $k$-tensor $B$
$$
\mu(B)=\min \left\{ \sum_{i} |\alpha_i| : B=\sum_i \alpha_i
\chi(Z_i), \alpha_i \in \R \right\}
$$
where each $Z_i$ is a cylinder intersection and $\chi(Z_i)$ is its characteristic tensor.

In the two dimensional case, $\mu$ is very closely related to the
$\gamma_2$ norm \cite{LMSS07, LS07}.  Indeed, for matrices $B$ we
have $\mu(B)=\Theta(\gamma_2(B))$.

\begin{remark}
In our definition of $\mu$ above we chose to take
$\chi(Z_i)$ as $\01$ tensors.  One can alternatively take them to
be $\pm 1$ valued tensors---a form which is sometimes easier to bound---without
changing much.  One can show
\[
\mu(B) \ge \mu_{\pm 1}(B) \ge 2^{-k}\mu(B).
\]
where $B$ is a $k$-tensor and $\mu_{\pm 1}(B)$ is defined as above
with $\chi(Z_i)$ taking values from $\mp$.  In the matrix case,
$\mu_{\pm}$ is also known as the nuclear norm \cite{Jam87}.
\end{remark}
By Fact~\ref{fact:partition_det} we have the following.
\begin{theorem}
\label{th:detcc_mu} It holds that $D^k(A) \ge \log(\mu(A))$ for
every sign $k$-tensor $A$.
\end{theorem}

A public coin randomized protocol is simply a probability distribution over deterministic protocols.
This gives us the following fact:
\begin{fact}
A sign $k$-tensor $A$ satisfies $R^k_{\epsilon}(A) \le c$ if and
only if there are sign $k$-tensors $A'_i$ for $i=1,\ldots, \ell$
satisfying $D^k(A'_i) \le c$ and a probability distribution $(p_1,
\ldots, p_\ell)$ such that
$$
\|A-\sum_{i=1}^\ell p_i A'_i \|_{\infty} \le 2 \epsilon.
$$
\label{fact:decomposition}
\end{fact}
To lower bound randomized communication complexity we
consider an approximate variant of the cylinder intersection norm.

\begin{definition}[Approximate cylinder intersection norm]
Let $A$ be a sign $k$-tensor, and $\alpha \ge 1$.  We define the
$\alpha$-approximate cylinder intersection norm as
$$
\mu^{\alpha}(A)=\min_{B} \{\mu(B): 1 \le A \circ B \le \alpha\}
$$
In words, we take the minimum of the cylinder intersection norm
over all tensors $B$ which are signed as $A$ and have entries with
magnitude between 1 and $\alpha$. Considering the limiting case as
$\alpha \rightarrow \infty$ motivates the definition
$$
\mu^{\infty}(A)=\min_{B} \{\mu(B): 1 \le A \circ B \}
$$
\label{def:acidn}
\end{definition}
One should note that $\mu^{\alpha}(A)\le \mu^{\beta}(A)$ for $1\le
\beta \le \alpha$.

The following theorem is an immediate consequence of the definition 
of the approximate cylinder intersection norm and \factref{fact:decomposition}.
\begin{theorem}
\label{th:cc_gamma} Let $A$ be a sign $k$-tensor, and $0 \le
\epsilon < 1/2$.  Then
$$
R^k_\epsilon(A) \ge \log(\mu^{\alpha}(A)) -\log(\alpha_\epsilon)
$$
where $\alpha_\epsilon=1/(1-2\epsilon)$ and $\alpha \ge \alpha_\epsilon$.
\end{theorem}

\begin{proof}
Let $p_i$ and $A'_i$ for $1 \le i \le \ell$ be as in
\factref{fact:decomposition}.  We take
$$
B=\frac{1}{1-2\epsilon} \sum_{i=1}^\ell p_i A'_i.
$$
Notice that $1 \le B \circ A \le \alpha_{\epsilon}$, and hence by
\defref{def:acidn}
$$
\mu^{\alpha_\epsilon}(A) \le \mu(B).
$$
Employing the fact that $\mu$ is a norm and \thmref{th:detcc_mu},
we get
\begin{align*}
\mu(B) &\le \frac{1}{1-2\epsilon}\sum_i p_i \mu(A'_i) \\
&\le \frac{1}{1-2\epsilon}\sum_i p_i 2^{D^k(A'_i)}\\
&\le \frac{2^{R^k_\epsilon(A)}}{1-2\epsilon}.
\end{align*}
\end{proof}

The nondeterministic complexity of a sign $k$-tensor $A$, denoted
$N^k(A)$, is the logarithm of the minimum cardinality of a set of cylinder
intersections $\{Z_i\}$ such that every entry of $A$ with value $-1$ is
covered by some $Z_i$, and no entry of $A$ with value $1$ is covered
by $Z_i$.  Notice that if $\{Z_i\}$ is such a covering of $A$, then letting
$B=-\sum \chi(Z_i)$ we have $1 \le A \circ (2B+J) < \infty$ where $J$ is the
all one tensor.  As $J$ is itself a cylinder, we have $\mu(J)=1$, which gives the following.

\begin{theorem}[folklore]
For a sign $k$-tensor $A$,
$$
N^k(A) \ge \log \frac{\mu^\infty(A)-1}{2}
$$
\end{theorem}

As we shall see in \secref{sec:discrepancy_method}, $\mu^\infty$ is exactly
the discrepancy method, which explains why the discrepancy method cannot
show good lower bounds on disjointness, or indeed any function with low
nondeterministic or co-nondeterministic communication complexity.

\subsection{Employing duality}
\label{sec:employing_duality}

We now have a quantity, $\mu^{\alpha}(A)$, which can be used to
prove lower bounds on randomized communication complexity in the
number-on-the-forehead model.  As this quantity is defined in
terms of a minimization, however, it seems in itself a difficult
quantity to bound from below.

In this section, we employ the duality theory of linear
programming to find an equivalent formulation of $\mu^{\alpha}(A)$
in terms of a maximization problem. This makes the task of proving
lower bounds for $\mu^{\alpha}(A)$ much easier, as the $\forall$
quantifier we had to deal with before is now replaced by an
$\exists$ quantifier.

As it turns out, in order to prove lower bounds on
$\mu^{\alpha}(A)$ we will need to understand the dual norm of
$\mu$, denoted $\mu^*$. The standard definition of a dual norm is
\[
\mu^*(Q) = \max_{B:\mu(B) \le 1} |\braket{B}{Q}|,
\]
for any tensor $Q$. Since the unit ball of $\mu$ is the absolute convex hull of
the characteristic vectors of cylinder intersections, we can alternatively write
$$
\mu^*(Q)=\max_{Z} |\braket{Q}{\chi(Z)}|
$$
where the maximum is taken over all cylinder intersections $Z$.

It is instructive to compare this with the definition of discrepancy.
\begin{definition}[discrepancy]
Let $A$ be a sign $k$-tensor, and let $P$ be a probability
distribution on its entries.  The discrepancy of $A$ with respect
to $P$, written $\disc_P(A)$, is
$$
\disc_P(A)=\max_{Z} |\braket{A \circ P}{\chi(Z)}|
$$
where the maximum is taken over cylinder intersections $Z$.  
\end{definition}
Thus we see that $\disc_P(A)=\mu^*(A \circ P)$, and we can use existing
techniques for discrepancy to also upper bound $\mu^*$.  

As the dual of a dual norm is again the norm, we can write the $\mu$ norm as 
\begin{equation}
\mu(B)=\max_Q \frac{\braket{B}{Q}}{\mu^*(Q)}.
\label{dualdual}
\end{equation}
To prove our lower bounds, we will use an equivalent formulation of 
$\mu^\alpha$ in terms of the dual norm $\mu^*$.  
\begin{theorem}
Let $A$ be a sign tensor and $1\le \alpha <\infty$.
\begin{equation*}
\mu^{\alpha}(A)= \max_Q \frac{(1+\alpha) \braket{A}{Q} + (1-\alpha)\|Q\|_1}{2\mu^*(Q)}
\end{equation*}
When $\alpha=\infty$ we have
\begin{equation*}
\mu^{\infty}(A)= \max_{Q: A \circ Q \ge 0} \frac{\braket{A}{Q}}{\mu^*(Q)}
\end{equation*}
\label{mu:dual}
\end{theorem}

\begin{proof}
We can quite easily see that the left hand side is at least as large as the right hand side, 
which is all that is needed for proving lower bounds.  By \eqnref{dualdual} and the 
definition of $\mu^\alpha$ we have
\begin{equation*}
\mu^\alpha(A)=\min_{B: 1 \le A \circ B \le \alpha} \max_Q \frac{\braket{B}{Q}}{\mu^*(Q)}.
\end{equation*}
If we rewrite $Q$ as the sum of two parts, $Q^+$, satisfying $Q^+ \circ A \ge 0$ and 
$Q^-$ satisfying $Q^- \circ A <0$ then we can see that
\begin{equation*}
\mu^\alpha(A) \ge \max_{Q^+, Q^-} \frac{\braket{A}{Q^+}+\alpha \braket{A}{Q^-}}{\mu^*(Q^+ +Q^-)}
\end{equation*}
It is now straightforward to verify that this expression can be 
reworked into the form given above in the two cases $1 \le \alpha < \infty$ and 
$\alpha=\infty$.  

To see that this inequality holds with equality, we write $\mu^\alpha$ as a linear program 
and then use duality to derive the dual expression given in the theorem.  As it is easy 
to check that the primal program is feasible with a finite optimum, by Slater's 
condition these primal and dual forms coincide with the same finite value. 

We treat the case $1 \le \alpha < \infty$ first. We can write
$\mu^{\alpha}(A)$ as a linear program as follows.  For each
cylinder intersection $Z_i$ let $X_i=\chi(Z_i)$.  Then
\begin{eqnarray*}
\mu^{\alpha}(A)&=&\min_{p,q} \sum_i p_i + q_i \\
\mbox{s.t.} &\mbox{ }& 1 \le \left(\sum_i (p_i - q_i) X_i \right) \circ A \le \alpha \\
& \mbox{ } & p_i , q_i \ge 0
\end{eqnarray*}

Taking the dual of this program in the straightforward way, we obtain
\begin{eqnarray*}
\mu^{\alpha}(A)&=& \max_Q \frac{(1+\alpha) \braket{A}{Q} + (1-\alpha)\|Q\|_1}{2} \\
   \mbox{s.t.} &\mbox{ }& |\braket{X_i}{Q}| \le 1, \mbox{ for all } X_i
\end{eqnarray*}

For $\alpha=\infty$ we get the same program as above without the
constraint $\left(\sum_i (p_i - q_i) X_i \right) \circ A \le
\alpha$.  Dualizing this program gives the desired result.
\end{proof}

Let us take a moment to compare our approach with that of Chattopadhyay 
and Ada.  They also use the approximation $\mu$ norm, but with an additive 
approximation factor rather than a multiplicative factor as we use.  More 
precisely, they use the measure 
$\mu^\epsilon(A)=\min_{B: \|A-B\|_\infty \le \epsilon} \mu(B)$.  The dual form of 
this measure has the form 
$$
\mu^\epsilon(A)=\max_Q \frac{\braket{A}{Q}-\epsilon\|Q\|_1}{\mu^*(Q)}.
$$

Chattopadhyay and Ada directly derive that this dual expression is a lower 
bound on multiparty {\em distributional} communication complexity.  Yao's 
characterization of randomized complexity in terms of distributional 
complexity \cite{Yao83} then gives that it is also a lower bound on randomized
communication complexity.  They do not mention the primal definition of $\mu^{\alpha}$,
but other than that, their proof is similar in structure to ours.  For our proof 
we do not use Yao's principle but apply duality directly on the measure $\mu$
rather than on the complexity class itself.    

While our presentation through the primal version of the $\mu$ norm is perhaps 
not as familiar as that via distributional complexity, we feel it does have advantages.
First of all, this discussion holds quite generally: 
for any norm $\Phi$ one can show using the separation theorem that the approximation 
version $\Phi^\alpha$ has a dual characterization analogous to that in \thmref{mu:dual}.
Second, we feel that the primal definition of $\mu^\alpha$ 
arises very naturally and gives insight into the origin of the dual formulation---we do not 
have to guess this formula but can derive it.  Finally, 
it is interesting to note that the primal and dual formulations are {\em equivalent}.  This
means that we do not lose anything in considering the more convenient dual formulation 
for proving lower bounds.

\subsection{The discrepancy method}
\label{sec:discrepancy_method}
Virtually all lower bounds in the general number-on-the-forehead model have used the
discrepancy method.  Let $A$ be a sign tensor, and recall the definition of $\disc_P(A)$
from Section~\ref{sec:employing_duality}.  Let $\disc(A)=\min_P \disc_P(A)$, where the 
minimum is taken over all probability distributions $P$.  The discrepancy method turns out to 
be equivalent to $\mu^{\infty}(A)$.
\begin{theorem}
$$
\mu^{\infty}(A)=\frac{1}{\disc(A)}.
$$
\end{theorem}

\begin{proof}
By \thmref{mu:dual}, for every sign tensor $A$
\begin{align*}
\mu^{\infty}(A)= \max_{Q \circ A \ge 0} \left\{ \braket{A}{Q} :
\mu^*(Q) \le 1\right\}
\end{align*}
We can rewrite this as
\begin{align*}
\mu^{\infty}(A)=\max_{Q \circ A \ge 0}
\frac{\braket{A}{Q}}{\mu^*(Q)} =\max_{P: P\ge 0}
\frac{\braket{A}{A \circ P}}{\mu^*(A \circ P)}
\end{align*}
As both numerator and denominator are homogeneous, we have
\begin{align*}
\mu^{\infty}(A)&=\max_{\substack{P: P\ge 0 \\ \|P\|_1=1}}
\frac{\braket{A}{A \circ P}}{\mu^*(A \circ P)}
=\max_{\substack{P:P\ge 0 \\ \|P\|_1=1}}\frac{1}{\mu^*(A \circ P)} \\
&= \frac{1}{\disc(A)}.
\end{align*}
\end{proof}

\section{Techniques to bound $\mu^*(Q)$}
\label{sec:tech_bound_mu_dual}

In the last section, we saw that to bound the randomized
number-on-the-forehead communication complexity of a sign tensor
$A$, it suffices to find a tensor $Q$ such that $\braket{A}{Q}$ is
large and $\mu^*(Q)$ is small.  The first quantity is relatively
simple and is in general not too hard to compute.  Upper bounding
$\mu^*(Q)$ is more subtle.  In this section, we review some
techniques for doing this.

In upper bounding the magnitude of the largest eigenvalue of a
matrix $B$, a common thing is to consider the matrix $BB^T$, and
use the fact that $\sn{B}^2 \le \|BB^T\|$. We will try to do a
similar thing in upper bounding $\mu^*$.  In analogy with $BB^T$
we make the next definition.  Here and in what follows all expectations are 
taken with respect to the uniform distribution.
\begin{definition}[Contraction product]
Let $B$ be a $k$-tensor with entries indexed by elements from $X_1
\times \ldots \times X_k$. We define the contraction product of
$B$ along $X_1$, denoted $B \bullet_1 B$, to be a $2(k-1)$-tensor
with entries indexed by elements from $X_2 \times X_2 \times
\ldots \times X_k \times X_k$. The $x_2,x_2', \ldots, x_k, x_k'$
entry is defined to be
\begin{equation*}
B\bullet_1 B[x_2,x_2', \ldots, x_k, x_k']= \Ex_{x_1} \left[
\prod_{y_2 \in \{x_2, x_2'\}, \ldots, y_k \in \{x_k, x_k'\}}
B[x_1, y_2,\ldots, y_k] \right]
\end{equation*}
The contraction product may be defined along other dimensions mutatis mutandis.
\end{definition}

Notice that when $B$ is a $m$-by-$n$ matrix $B \bullet_1 B$
corresponds to $(1/m)BB^T$. In analogy with the fact that
$\sn{B}^2 \le m \sn{B \bullet_1 B}$, the next lemma gives a
corresponding statement for the $\mu^*$ norm and $k$-tensors. This
lemma originated in the work of Babai, Nisan, and Szegedy
\cite{BNS89} (see also \cite{Chu90, Raz00}) and all lower bounds
on the general model of randomized \NOF complexity use some
version of this lemma.  The particular statement we use is from
Chattopadhyay \cite{Cha07}. 

\begin{lemma}
Let $B$ be a $k$-tensor.  Then
$$
\left(\frac{\mu^*(B)}{\size(B)}\right)^{2^{k-1}} \le \
\frac{\mu^*(B \bullet_1 B)}{\size(B \bullet_1 B)} \le \ \Ex[|B
\bullet_1 B|]
$$
\label{claim:contraction}
\end{lemma}

\begin{proof}
The second inequality follows since $\mu^*(X) \le \|X\|_{1}$ for
any real tensor $X$. The first inequality is standard, and follows
by applying the Cauchy-Schwarz inequality repeatedly $k-1$ times.
\end{proof}

\subsection{Example: Hadamard tensors} We give an example to show how
\lemref{claim:contraction} can be used in conjunction with our $\mu$ method.
Let $H$ be a $N$-by-$N$ Hadamard matrix.  We show that $\mu^{\infty}(H) \ge \sqrt{N}$.  Indeed,
simply let the witness matrix $Q$ be $H$ itself.  Incidentally, this corresponds to taking the uniform
probability distribution in the discrepancy method.  With this choice we clearly have $H \circ Q \ge 0$,
and so
$$
\mu^{\infty}(H) \ge \frac{\braket{H}{H}}{\mu^*(H)}=\frac{N^2}{\mu^*(H)}
$$
Now we bound $\mu^*(H)$ using \lemref{claim:contraction} which gives:
\begin{align*}
\mu^*(H)^2 \le N^4 \ \Ex[|H \bullet_1 H|] =N^3
\end{align*}
as $H \bullet_1 H$ has nonzero entries only on the diagonal, and these entries are of magnitude
one.

Ford and G{\'a}l \cite{FG05} extend the notion of matrix orthogonality to tensors, defining what they call
Hadamard tensors.
\begin{definition}[Hadamard tensor]
Let $H$ be a sign $k$-tensor.  We say that $H$ is a Hadamard tensor if
$$
(H \bullet_1 H)[x_2, x_2', \ldots, x_k, x_k']=0
$$
whenever $x_i \ne x_i'$ for all $i=2, \ldots, k$.
\end{definition}

The simple proof above for Hadamard matrices can be easily extended to Hadamard tensors:
\begin{theorem}[Ford and G{\'a}l \cite{FG05}]
Let $H$ be a Hadamard $k$-tensor of side length $N$.  Then
$$
\mu^{\infty}(H) \ge \left(\frac{N}{k-1}\right)^{1/2^{k-1}}
$$
\end{theorem}

\begin{proof}
We again take the witness $Q$ to be $H$ itself.  This clearly satisfies $H \circ Q \ge 0$, and so
$$
\mu^{\infty}(H) \ge \frac{\braket{H}{H}}{\mu^*(H)} = \frac{N^k}{\mu^*(H)}
$$
It now remains to upper bound $\mu^*(H)$ which we do by \lemref{claim:contraction}.  This
gives us
$$
\mu^*(H)^{2^{k-1}} \le N^{k2^{k-1}} \  \Ex[|H \bullet_1 H|]
$$
The ``Hadamard'' property of $H$ lets us easily upper bound $\Ex[|H \bullet_1 H|]$.  Note that
each entry of $H \bullet_1 H$ is of magnitude at most one, and the probability of a non-zero entry
is at most
$$
\Pr[\vee_{i=2}^k (x_i=x_i')] \le \frac{k-1}{N}
$$
by a union bound.  Hence, we obtain
$$
\mu^*(H)^{2^{k-1}} \le (k-1) \frac{N^{k2^{k-1}}}{N}.
$$
Putting everything together, we have
$$
\mu^{\infty}(H) \ge \left(\frac{N}{k-1}\right)^{1/2^{k-1}}
$$
\end{proof}

\begin{remark}
By doing a more careful inductive analysis, Ford and G{\'a}l obtain this result without the $k-1$ term in
the denominator.  They also construct explicit examples of Hadamard tensors.
\end{remark}

\section{Lower bounds on $\mu^{\alpha}$ for pattern tensors}
\label{sec:lower_bounds_on_mu_alpha}

In Section~\ref{sec:dual_polynomials} we describe a key lemma
which relates the approximate polynomial degree of $f$ to the
existence of a hard input ``distribution'' for $f$.  This will
only truly correspond to a distribution in the case of
discrepancy---otherwise it can take on negative values.  This lemma
was first used in the context of communication complexity by Sherstov
\cite{She08a} and independently by Shi and Zhu \cite{SZ07}.

In \secref{sec:pattern_tensors} we use this distribution, together
with the machinery developed in \secref{sec:tech_bound_mu_dual} to
prove lower bounds on a special kind of tensors, named pattern tensors.
The application to disjointness appears in Section~\ref{sec:lower_bound_disj}.

\subsection{Dual polynomials}
\label{sec:dual_polynomials}
We define approximate degree in a slightly non-standard way to more smoothly handle both
the bounded $\alpha$ and $\alpha=\infty$ cases.
\begin{definition}
Let $f: \{-1,+1\}^n \rightarrow \mp$.  For $\alpha \ge 1$ we say that a function $g$ gives an $\alpha$-approximation to $f$ if $1 \le g(x) f(x) \le \alpha$ for all $x \in \{-1,+1\}^n$.  Similarly we say that
$g$ gives an $\infty$-approximation to $f$ if $1 \le g(x) f(x)$ for all $x \in \{-1,+1\}^n$.  We let the
$\alpha$-approximate degree of $f$, denoted
$\deg_\alpha(f)$, be the smallest degree of a function $g$ which gives an $\alpha$-approximation
to $f$.
\end{definition}

\begin{remark}
In a more standard scenario, one is considering a 0/1 valued function $f$ and defines the
approximate degree as $\deg'_\epsilon(f)=\min\{\deg(g): \|f-g\|_{\infty} \le \epsilon\}$.  Letting
$f_\pm$ be the sign representation of $f$, one can see
that for $0\le \epsilon < 1/2$ our definition
is equivalent to the standard one in the following sense:
$\deg'_{\epsilon}(f)=\deg_{\alpha_\epsilon}(f_\pm)$ where
$\alpha_{\epsilon}=\frac{1+2\epsilon}{1-2\epsilon}$.
\end{remark}

For a fixed natural number $d$, let $\alpha_d(f)$ be the smallest
value of $\alpha$ for which there is a degree $d$ polynomial which
gives an $\alpha$-approximation to $f$.  Notice that $\alpha_d(f)$
can be written as a linear program. Namely, let
$B(n,d)=\sum_{i=0}^d \binom{n}{i}$, and $W$ be a
$2^n$-by-$B(n,d)$ incidence matrix, with rows labelled by strings
$x\in \{-1,+1\}^n$ and columns labeled by monomials of degree at most
$d$.  We set $W(x, m)=m(x)$, where $m(x)$ is the evaluation of the
monomial $m$ on input $x$. Then
$$
\alpha_d(f)=\min_y \{\|Wy\|_{\infty} : 1 \le W y \circ f\}
$$
If this program is infeasible with value $\alpha$---that is, if there is no degree $d$ polynomial which
gives an $\alpha$-approximation to $f$---then the feasibility of the dual of this program will give us
a ``witness'' to this fact.  We refer to this witness as a dual polynomial for $f$.  
It is this witness that we will use to construct a tensor $Q$ which witnesses that
$\mu^{\alpha}$ is large.

\begin{lemma}
$$
\alpha_d(f)=\max_v \left\{\frac{1+ \braket{v}{f}}{1-\braket{v}{f}} : \|v\|_1=1, v^T W=0 \right\}
$$
\end{lemma}

\begin{proof}
Follows from duality theory of linear programming.
\end{proof}

\begin{corollary}[Sherstov Corollary 3.3.1 \cite{She08a}, Shi-Zhu Section 3.1 \cite{SZ07}]
Let $f:\{-1,+1\}^n \rightarrow \R$ and let $d=\deg_{\alpha}(f)$.
Then there exists a function $v: \{-1,+1\}^n \rightarrow \R$ such that
\begin{enumerate}
  \item $\braket{v}{\chi_T}=0$ whenever $|T| \le d$.
  \item $\|v\|_1=1$.
  \item $\braket{v}{f} \ge \frac{\alpha -1}{\alpha+1}$.
\end{enumerate}
When $\alpha=\infty$, there is a function $v: \{-1,+1\}^n
\rightarrow \R$ satisfying items (1), (2), and such that $v(x)
f(x) \ge 0$ for all $x \in \{-1,+1\}^n$.
\label{cor:dual_poly}
\end{corollary}

\v{S}palek \cite{Spa08} has given an explicit construction of an optimal dual polynomial for the
$\OR$ function.  For our analysis, however, we only make use of the properties guaranteed
by \corref{cor:dual_poly}.

\subsection{Pattern Tensors}
\label{sec:pattern_tensors} We define a natural generalization of
the pattern matrices of Sherstov \cite{She07} to the tensor case.
We use a slightly different definition of pattern tensors than
that of Chattopadhyay \cite{Cha07} to allow the reduction to
disjointness.

Let $\phi: \{-1,+1\}^m \rightarrow \R$ be a function and $M$ a natural number.  We define a
$k$-dimensional pattern tensor $A_{k,M, \phi}$ as follows.  Let $x \in \{-1,+1\}^{mM^{k-1}}$.  We
view $x=(x^1, \ldots, x^m)$ as consisting of $m$ many blocks, where each
$x_i \in \{-1,+1\}^{M^{k-1}}$ can be viewed as a $k-1$ dimensional tensor of side length $M$.  We
further let
$y_i \in [M]^m$ for each $i=1, \ldots, k-1$ and view each $y_i=(y_i[1], \ldots, y_i[m])$ as
consisting of $m$-blocks where $y_i[j] \in [M]$ is an index into a side of $x^i$.  Now define
$$
A_{k,m,\phi}[x, y_1, \ldots, y_{k-1}]=\phi(x^1[y_1[1], \ldots, y_{k-1}[1]], \ldots,
x^m[y_1[m], \ldots, y_{k-1}[m]]).
$$
Note that $\size(A_{k,M,\phi})=2^{mM^{k-1}}M^{m(k-1)}$.  We will often use the abbreviation
$\bar y=(y_1, \ldots, y_{k-1})$.  A nice property of pattern tensors is that every $m$-bit string
$z$ appears as input to $\phi$ an equal number of times, over all choices of $x, \bar y$.

The key lemma about pattern tensors is given next.  Such a lemma was first shown
by Chattopadhyay \cite{Cha07}.  Chattopadhyay and Ada \cite{CA08} also show a statement
similar to this one.
\begin{lemma}
Let $A$ be a $(k,M, c \cdot \phi)$ pattern tensor, where $c=2^m \size(A)^{-1}$.
Suppose that $\phi$ satisfies $\ell_1(\phi)= 1$ and
$\hat \phi_T=0$ for all sets $T \subseteq [m]$ with $|T| \le d$.  Then
$$
\mu^*(A) \le 2^{-d}
$$
provided that $M \ge 2e(k-1)2^{2^{k-1}}m/d$.
\label{lem:main}
\end{lemma}

\begin{proof}
The idea of the proof will be to bound $\Ex[|A \bullet_1 A|]$ and apply \lemref{claim:contraction}
to obtain an upper bound on $\mu^*(A)$.  For a string $\ell \in \01^{k-1}$ we use
the abbreviation $\bar y^\ell=(y_1^{\ell_1}, \ldots, y_{k-1}^{\ell_{k-1}})$.  In particular,
$\bar y^0=(y_1^0, \ldots, y_{k-1}^0)$ and $\bar y^1=(y_1^1, \ldots, y_{k-1}^1)$.
\begin{align}
\Ex[|A \bullet_1 A|]&=
\left(\frac{2^m}{\size(A)}\right)^{2^{k-1}}
\Ex_{\bar y^0, \bar y^1} \left|
\Ex_x \prod_{\ell=0}^{2^{k-1}-1}
\sum_{T \subseteq [m]} \hat \phi(T) \prod_{i \in T} x^i[y_1^{\ell_1}[i], \ldots, y_{k-1}^{\ell_{k-1}}[i]]
\right| \\
&\le \frac{1}{\size(A)^{2^{k-1}}} \Ex_{\bar y^0, \bar y^1}
\sum_{\substack{T_0, \ldots, T_{2^{k-1}-1} \\ |T_\ell | > d}}
\prod_{i=1}^m \left| \Ex_{x^i} \prod_{\substack{\ell \in \01^{k-1} \\ i \in T_\ell}}
x^i[y_1^{\ell_1}[i], \ldots, y_{k-1}^{\ell_{k-1}}[i]] \right|.
\label{eq:expectation}
\end{align}
Here we have used the fact that $\hat \phi(T) \le 2^{-n} \ell_1(\phi)=2^{-n}$.

We now develop a sufficient condition in terms of $\bar y^0, \bar y^1$ and
$T_0, \ldots, T_{2^{k-1}-1}$, for the
product of expectations over $x^i$ to be zero.  We say that
$\bar y^0, \bar y^1$ select a nondegenerate cube in
position $i$ if $y_j^0[i] \ne y_j^1[i]$ for all $j=1, \ldots, k-1$.  The reason for this terminology is
that in this case $(y_1^{\ell_1}[i], \ldots, y_{k-1}^{\ell_{k-1}}[i])$ define $2^{k-1}$ distinct points over
$\ell \in \01^{k-1}$.  If this is not the case, we say that $\bar y^0, \bar y^1$ select a degenerate
cube in position $i$.

Notice that if $\bar y^0, \bar y^1$ select a nondegenerate cube in position $i \in [m]$ and
$i \in T_\ell$ for some $\ell \in \01^{k-1}$ then
$$
\Ex_{x^i} \prod_{\substack{\ell \in \01^{k-1} \\ i \in T_\ell}}
x^i[y_1^{\ell_1}[i], \ldots, y_{k-1}^{\ell_{k-1}}[i]]=0.
$$

We will now upper bound the probability over the choice of $\bar y^0, \bar y^1$ and
$T_0, \ldots, T_{2^{k-1}-1}$ that this does not happen.
Suppose that $\bar y^0, \bar y^1$ select $g$ many degenerate cubes.  By the above
reasoning the number of sets $T_0, \ldots, T_{2^{k-1}-1}$
which lead to a nonzero expectation is at most
$$
\left(\sum_{r=d+1}^g \binom{g}{r}\right)^{2^{k-1}} \le 2^{g2^{k-1}}.
$$

Now we bound the probability that $\bar y^0, \bar y^1$ select $g$
many degenerate cubes.  The probability that $y_j^0[i]=y_j^1[i]$
is $1/M$.  Thus by a union bound, the probability that a single
cube is degenerate is at most $(k-1)/M$.  Finally, as each index
is chosen independently, the probability of $g$ many degenerate
cubes is at most
$$
\binom{m}{g} \left(\frac{k-1}{M}\right)^g.
$$

Putting everything together we have
\begin{align*}
\Ex[|A \bullet_1 A|] &\le \frac{1}{\size(A)^{2^{k-1}}}
\sum_{g=d+1}^m \binom{m}{g} \left(\frac{k-1}{M}\right)^g 2^{g2^{k-1}} \\
&\le \frac{1}{\size(A)^{2^{k-1}}} \sum_{g=d+1}^m \left(\frac{e(k-1)2^{2^{k-1}}m}{dM}\right)^g \\
& \le \frac{2^{-d}}{\size(A)^{2^{k-1}}}
\end{align*}
provided that $M \ge 2e(k-1)2^{2^{k-1}}m/d$.
\end{proof}

\begin{remark}
Our analysis cannot be improved by much without using more explicit
information about the Fourier coefficients $\hat q(T)$ than given in
\corref{cor:dual_poly}.  Apart from removing the Fourier coefficients, the only inequality we have
used to arrive at \eqnref{eq:expectation} is to turn an absolute value of a sum into a sum of
absolute values.  When $\bar y^0, \bar y^1$ select a degenerate cube, the most likely case is
that it is what we call $1$-degenerate---that is $y_i^0[t]=y_i^1[t]$ for {\em exactly one} $1\le i \le k-1$.
If the degenerate cubes selected by $\bar y^0, \bar y^1$ are all $1$-degenerate, then one can
see that the only sets $\{T_\ell\}$ which lead to a nonzero expectation are ones where the sets $T_\ell$
come in pairs.  The number of such paired sets $\{T_\ell\}$ is not significantly smaller than the
upper bound we give; furthermore, in this case all Fourier coefficients will be taken to an
even power and so no cancellation occurs and the absolute value of the sum will be equal to the sum
of absolute values.
\end{remark}

With this lemma in hand, we can now show our main result, proving
a lower bound on $\mu^\alpha(A_{k,M,f})$ in terms of the
approximate degree of $f$.

\begin{theorem}
\label{th:deg_mu_general}
For a nonnegative integer $m$ and a Boolean function $f$ on $m$ variables, and an integer
$k \ge 2$
\[
\log \mu^{\alpha}(A_{k,M,f}) \ge  \deg_{\alpha_0}(f)/2^{k-1} + \log \frac{\alpha_0-\alpha}{\alpha_0+1},
\]
for every $1 \le \alpha < \alpha_0 < \infty$, provided $M \ge
2e(k-1)2^{2^{k-1}}m/\deg_{\alpha_0}(f)$.
\\
Furthermore,
\[
\log \mu^{\infty}(A_{k,M,f}) \ge \deg_{\infty}(f)/2^{k-1},
\]
provided $M \ge 2e(k-1)2^{2^{k-1}}m/\deg_\infty(f)$
\end{theorem}

\begin{proof}
For simplicity we will drop the subscripts and just write $A$ for $A_{k,M,f}$.  Recall
that
\begin{align*}
\mu^{\alpha}(A)&=\max_{Q:\|Q\|_1=1} \frac{(1+\alpha)\braket{A}{Q}+(1-\alpha)}{2\mu^*(Q)} \\
\mu^{\infty}(A)&=\max_{Q: Q\circ A \ge 0} \frac{\braket{A}{Q}}{\mu^*(Q)}.
\end{align*}

Let $q$ be the vector from \corref{cor:dual_poly} which witnesses that the $\alpha_0$-approximate
degree of $f$ is at least $d$.  We let $Q$ be the $(k,M,c \cdot q)$ pattern tensor where
$c=2^m/\size(A)$.  This choice of normalization implies that $\|Q\|_1=1$ as $\|q\|_1=1$.

First consider the case $1\le \alpha < \infty$.  Then we have
$\braket{q}{f} \ge (\alpha_0-1)/(\alpha_0+1)$, and so
$\braket{A}{Q} \ge (\alpha_0-1)/(\alpha_0+1)$.  This allows us to bound $(1/2)$ the term in the
numerator of $\mu^{\alpha}(A)$ as follows:
\begin{align*}
\frac{(1+\alpha)\braket{A}{Q} + (1-\alpha)}{2} &\ge
\frac{\alpha_0-\alpha}{\alpha_0+1}.
\end{align*}

In the case $\alpha=\infty$, observe that $Q$ inherits the property $Q \circ A \ge 0$ as $q \circ f \ge 0$.
The fact that $q \circ f \ge 0$ together with $\|q\|_1=1$ gives $\braket{f}{q}=1$, which in turn
implies $\braket{A}{Q}=1$.

Let $d=\deg_{\alpha_0}(f)$ or $d=\deg_\infty(f)$, respectively.
As $q$ has no nonzero Fourier coefficients of degree less than $d$ by
\corref{cor:dual_poly}, we can apply Lemma~\ref{lem:main} to give
$$
\mu^*(Q) \le 2^{-d},
$$
under the assumption that $M \ge 2e(k-1)2^{2^{k-1}}m/d$.  The statement now follows from
\lemref{claim:contraction}.
\end{proof}

\section{Applications}
In this section, we apply Theorem~\ref{th:deg_mu_general}
to prove lower bounds on the $k$-party number-on-the-forehead randomized
communication complexity of disjointness.  Then we formally state the implications
this result has for proof systems via the results of Beame, Pitassi, and Segerlind \cite{BPS06}.

\subsection{A lower bound for disjointness}
\label{sec:lower_bound_disj}
Let $\OR_n: \{-1,+1\}^n \to \{-1,+1\}$ be the OR function on $n$ bits, and let
$\DISJ_{k,n}:(\{-1,+1\}^n)^k \rightarrow \{-1,+1\}$ be defined as
$\DISJ_{k,n}(x_1,\ldots,x_k) = -\OR_n(x_1 \wedge x_2 \ldots \wedge x_k)$.

By embedding a pattern tensor into the tensor $\DISJ_{k,n}$, we can get the following
lower bound.
\begin{corollary}
$$
R_{1/4}(\DISJ_{k,n})=\Omega \left(\frac{n^{1/(k+1)}}{2^{2^{k}}}\right)
$$
\label{cor:disj}
\end{corollary}

\begin{proof}
The idea of the proof will be to embed an appropriate pattern tensor into $\DISJ_{k,n}$ and
apply \thmref{th:deg_mu_general}.
Let $c_k=5e(k-1)2^{2^{k-1}}$.  As Nisan and Szegedy have shown $\deg_3(\OR_n) \ge \sqrt{n/6}$,
we wish to define integers $m, M$ such that $M \ge c_k \sqrt{m}$ and
$mM^{k-1} \le n$.  To this end, let $m=\floor{\tfrac{n}{(2c_k)^{k-1}}}$ and
$M = c_k \ceil{\sqrt{m}\;}$.  Let $n'=mM^{k-1}$.  One can easily check that $n' \le n$.

We will now see that the pattern tensor $(k,M, \OR_m)$ is a subtensor of
$\OR_{n'}(x_1 \wedge \ldots \wedge x_k)$.  This will then give the result
by the obvious reduction to $\DISJ_{k,n}$.

Let $A$ be the $(k,M, \OR_m)$ pattern tensor.  Recall that
$$
A[x, y_1, \ldots, y_{k-1}]=\OR_m(x^1(y_1[1], \ldots, y_{k-1}[1]), \ldots, x^m[y_1[m], \ldots, y_{k-1}[m]]),
$$
where each $y_j[i] \in [M]$, and $x^j$ is a $k-1$ dimensional tensor of side length $M$.
To each $y_j[i]$ we associate a $k-1$ tensor $z_j^i$ of side length $M$, where
$z_j^i[t_1, \ldots, t_{k-1}]=1$ if and only if $t_j=y_j[i]$.  In this way,
$x^1[y_1[1], \ldots, y_{k-1}[1]]=\OR_{M^{k-1}}(x^1 \wedge z_1^1 \wedge \ldots \wedge z_{k-1}^1)$.
Letting $z_j=(z_j^1, \ldots, z_j^m)$ we have
$$
\OR_{n'}(x_1 \wedge z_1 \ldots \wedge z_{k-1})=
\OR_m( \OR_{M^{k-1}}(x_1^1 \wedge z_1^1 \wedge \ldots \wedge z_{k-1}^1), \ldots,
\OR_{M^{k-1}}(x_1^m \wedge z_1^m \wedge \ldots \wedge z_{k-1}^m)).
$$
This shows that $A$ is a subtensor of $-\DISJ_{k,n'}$.
The result now follows from \thmref{th:deg_mu_general} and \thmref{th:cc_gamma}.
\end{proof}


\begin{remark}
Note that a statement similar to that of \corref{cor:disj}
can be proved for any symmetric function, not just $\OR$.  But for
some functions (e.g.\ threshold functions with threshold a constant fraction of $n$) much better
bounds can be proved by reduction to inner product. For this reason, we do not include the
general statement here.
\end{remark}

\subsection{Proof systems}
\label{sec:proof_systems}
In this section we formally define the proof systems discussed in the introduction, and
the lower bounds which follow from our results on disjointness.

A $k$-threshold formula is a formula of the form $\sum_j \gamma_j m_j \ge t$, where $t, \gamma_j$ are
integers, and each $m_j$ is a monomial over variables $x_1, \ldots, x_n$.  The size of a $k$-threshold
formula is the sum of the sizes of $\gamma_j$ and $t$, written in binary.  For $k$-threshold formulas
$f_1, f_2, g$, we say that $g$ is {\em semantically entailed} by $f_1$ and $f_2$ if every 0/1 assignment
to $x_1, \ldots, x_n$ that satisfies both $f_1$ and $f_2$ also satisfies $g$.

Let $\phi$ be an unsatisfiable CNF formula with variables $x_1, \ldots, x_n$.  For each clause of
$\phi$ we create a linear threshold formula which is satisfied if and only if the clause is.  We refer to
these clauses as {\em axioms}.  We say that $\mathcal{P}$ is a $\Th(k)$ refutation of $\phi$ if
\begin{itemize}
  \item $\mathcal{P}$ is a sequence $L_1, \ldots, L_t$ of $k$-threshold formulas.
  \item Each formula $L_j$ is either an axiom or is semantically entailed by formulas $L_i, L_{i'}$
  with $i, i' < j$.
  \item The final formula $L_t$ is $0 \ge 1$.
\end{itemize}
The size of $\mathcal{P}$ is the sum of the sizes of $L_1, \ldots, L_t$.  We say that
$\mathcal{P}$ is {\em tree-like} if the underlying directed acyclic graph representing the
implication structure of the proof is a tree.

We are now ready to state the connection of \cite{BPS06} between the \NOF complexity of
disjointness and the size of $\Th(k)$ proofs.

\begin{theorem}[Beame, Pitassi, and Segerlind \cite{BPS06}]
Let $k \ge 2$ be a constant.  For every $n$, there is a CNF formula $\phi$ on $n$ variables such
that the size of any $\Th(k-1)$ refutation of $\phi$ is at least
$$
\exp\left(\Omega\left(\frac{R_{1/4}^k(\DISJ_{k,m})}{\log n}\right)^{1/3} \right).
$$
where $m=\tfrac{n^{2/3}}{2\log n}$. 
\end{theorem}

Substituting the bounds from \corref{cor:disj} we obtain the following.
\begin{corollary}
Let $k \ge 2$ be a constant.  For every $n$ there is a CNF formula $\phi$ over $n$ variables
which requires $\Th(k-1)$ refutation proofs of size
$$
\exp\left(\Omega\left(\frac{n^{2/(9k+9)}}{(\log n)^{4/9} \ 2^{2^{k}/3}}\right)\right).
$$
\end{corollary}

\section*{Acknowledgments}
We greatly benefited during the course of this work from comments and conversations with
Paul Beame, Harry Buhrman, Mike Saks, Gideon Schechtman, Nate Segerlind, Sasha Sherstov,
Robert \v{S}palek, Emanuele Viola, Fengming Wang, Avi Wigderson, and Ronald de Wolf.  
We would also like to thank the anonymous referees for their many suggestions.


\end{document}